\documentclass[11pt]{article}

  \usepackage{ifpdf}

  \makeatletter
  \let\@font@warningori\@font@warning
  \newcommand\shutup{\def\@font@warning##1{}}
  \newcommand\youcanspeak{\let\@font@warning\@font@warningori}
  \makeatother

  \usepackage{amsmath,amssymb,amsbsy,bm,latexsym,enumerate,mathrsfs}

  \shutup
  \usepackage{fourier}
  \youcanspeak
  \usepackage[scaled=0.875]{helvet}
  
  \usepackage[latin1]{inputenc}
  \usepackage[german,dutch,english]{babel}

\usepackage{amsthm}

\usepackage{graphicx,epsfig,color}
\usepackage{tikz}

\usepackage{anysize}

\usepackage{subfigure}

\providecommand{\norm}[1]{\lVert#1\rVert}

\providecommand{\classNP}{\mathrm{NP}}

\usepackage[textsize=footnotesize]{todonotes}

\newcommand {\R}	  {\mathbb{R}}

\renewcommand{\epsilon}{\varepsilon}

\renewcommand{\leq}{\leqslant}
\renewcommand{\geq}{\geqslant}

\newcommand{\M}{\mathcal{M}}
\newcommand{\I}{\mathcal{I}}
\newcommand{\loss}{\mathrm{loss}}

\theoremstyle{plain}
\newtheorem{theorem}{Theorem}
\newtheorem{lemma}[theorem]{Lemma}

\theoremstyle{definition}

\newtheorem{remark}{Remark}

\title{Max-sum diversity via convex programming} 
\author{
  Alfonso Cevallos\footnote{\'Ecole Polytechnique Fédérale de Lausanne (EPFL), Switzerland. \texttt{alfonso.cevallosmanzano@epfl.ch}} \quad {\small }  \quad  Friedrich Eisenbrand\footnote{\'Ecole Polytechnique Fédérale de Lausanne (EPFL), Switzerland. \texttt{friedrich.eisenbrand@epfl.ch}} \quad {\small } \quad Rico Zenklusen\footnote{Swiss Federal Institute of Technology in Zurich (ETH Zurich), Switzerland. \texttt{ricoz@math.ethz.ch}} \\[.1cm]
}\date{\today }

\providecommand{\OPT}{\mathrm{OPT}}
\providecommand{\MSD}{\mathrm{MSD}}

\providecommand{\M}{\mathscr{M}}
\providecommand{\I}{\mathscr{I}}

\begin{document}
\maketitle

\begin{abstract}
  \noindent
  \emph{Diversity maximization} is an important concept in 
  information retrieval, computational geometry and operations research. Usually, 
  it is a variant of the following problem:  Given a ground set,
  constraints, and a function $f(\cdot)$ that measures  diversity of a subset, the task is to
  select a feasible subset $S$ such that $f(S)$ is maximized. 
  The \emph{sum-dispersion} function
  $f(S) = \sum_{x,y \in S} d(x,y)$, 
  which is the sum of the pairwise distances in $S$,
  is in this context a  prominent diversification measure.  The
  corresponding diversity maximization is the \emph{max-sum} or
  \emph{sum-sum diversification}.  Many
  recent results deal with the design of constant-factor approximation
  algorithms of diversification problems involving sum-dispersion
  function under a matroid constraint. 

  In this paper, we present a PTAS for the max-sum diversification
  problem under a matroid constraint for distances $d(\cdot,\cdot)$
  of \emph{negative type}. Distances of negative type are, for
  example, metric distances stemming from the $\ell_2$
  and $\ell_1$
  norm, as well as the cosine or spherical, or Jaccard distance which
  are popular similarity metrics in web and image search.

  Our algorithm is based on techniques developed in geometric
  algorithms like metric embeddings and convex optimization. We show
  that one can compute a fractional solution of the usually non-convex
  relaxation of the problem which yields an upper bound on the optimum
  integer solution. Starting from this fractional solution, we employ
  a deterministic rounding approach which only incurs a small loss
  in terms of objective, thus leading to a PTAS.
  This technique can be applied to other previously studied variants
  of the max-sum dispersion function, including combinations of
  diversity with linear-score maximization,
  %or packing and matroid constraints,
  improving the previous constant-factor approximation
  algorithms.

% is not concave, our result is based on convex
%  programming. We show that one can, in polynomial time, compute a
%  feasible fractional point in the linear relaxation whose diversity
%  is an upper bound on the optimal value. Our main result follows then
%  from randomized rounding.
  
%  This technique can be applied to other variants of the max-sum
%  diversification problem including max-sum diversification over a
%  matroid constraint. In the latter case, our algorithm finds a
%  $(1+\epsilon)$-approximation
%  if $OPT = \Omega(k/\epsilon)$,
 % where $k$ is the cardinality of a basis of the matroid.

\end{abstract}

%\category{G.2}{Mathematics of Computing}{Discrete Mathematics}
%\keywords{Diameter of polyhedra, polyhedral graph, totally unimodular matrices, isoperimetric inequality}

\section{Introduction}
\label{sec:bound}

Diversification is an  important concept   in
many areas of computing such as information retrieval, computational geometry or optimization.  
When searching for news on a particular subject, for example, one is
usually confronted with several relevant search results. A news-reader
might be interested in news from various sources and viewpoints. Thus,
on the one hand, the articles should be relevant to his search and, on
the other hand, should be significantly diverse.

% How can this diversity be formalized? Gollapudi and Sharmar~\cite{gollapudi2009axiomatic} described several axioms that a diversity measure should satisfy and concluded that not all of those can be satisfied simultaneously. However, they suggested a measure of diversity which has many desirable features.

The so-called \emph{max-sum diversification} or \emph{max-sum dispersion} is a  diversity-measure  that has been subject of study in operations  research~\cite{hassin1997approximation,ravi1994heuristic,birnbaum2009improved,baur2001approximation} for a while  and it  is currently receiving  considerable attention in the information retrieval  literature~\cite{gollapudi2009axiomatic,bhattacharya2011consideration,borodin2012max}. It is readily described.  
Given a ground set $X$ together with a  distance function $d: X \times X \rightarrow \R_{\geq 0}$. The diversity, or dispersion, of a subset $S \subseteq X$ is  the sum of the pairwise distances 
\begin{displaymath}
  f(S) = \sum_{i,j \in S} d(i,j).  
\end{displaymath}

Since documents are often represented as vectors in a high-dimensional space,   their similarity is measured by norms in $\R^n$ and their induced distances, see, e.g., \cite{manning2008introduction, salton1986introduction}. Among the most frequent  distances are the ones induced by the $\ell_1$ and $\ell_2$ norm,  the \emph{cosine-distance} or the \emph{Jaccard distance}~\cite{Pekalska:2005:DRP:1197035}.  These norms are also used to measure similarity via much lower-dimensional bit-vectors stemming from sketching techniques~\cite{charikar2002similarity}. So, usually, the ground set $X$ is a finite set of vectors in $\R^d$ and $d(\cdot,\cdot)$ is a metric on $\R^d$ which makes diversity maximization  a \emph{geometric optimization problem}.

\bigskip 
\noindent 
Before we go on, we state the general version  of 
the  \emph{maximum sum diversification} or \emph{maximum sum dispersion} ($\MSD$) problem which generalizes many previously studied variants. It is in the focus of this paper, and is described by the following  quadratic integer programming problem:
\begin{equation}
  \label{eq:1}
  \begin{array}{ll}
    \text{maximize} &   x^T D x\\
    \\
    \text{subject to} & A x \leq b \\
    \\
                      & x_i \in \{0,1\} \, \text{ for }\,i \in [n],
  \end{array}
\end{equation}
where the symmetric matrix $D \in \R^{n\times n}$ represents  the distances between the points of a \emph{distance space} and  $Ax \leq b$ is a set of additional linear constraints where $A \in \R^{m\times n}$ and $b \in \R^m$.  
If $x \in \{0,1\}^n$ is the characteristic vector of the set $S \subseteq X$, then $x^TDx =  \sum_{i,j \in S} d(i,j)$. 
We recall the notion of a  distance space, see, e.g., \cite{DezaLaurent97}. It is a pair $(X,d)$
where $X$
is a finite set and $d(\cdot,\cdot)$
is the \emph{distance} function
$d:X\times X\rightarrow \R_{\geq 0}$. 
The function $d$
satisfies $d({i,i})=0$
and $d({i,j})=d({j,i})$
for all $i,j\in X$.
If in addition $d$
satisfies the triangle inequality $d({i,j})\leq d({i,k})+d({j,k})$
for all $i,j,k\in X$,
then $d$ is 
a \emph{(semi) metric} and $(X,d)$ a finite \emph{metric space}.

%\subsubsection*{Previous results}
%\label{sec:probl-stat-relat}

\bigskip 
\noindent 
We now mention some previous algorithmic work on diversity maximization with this objective function. 
For the case where $Ax \leq b$
represents one cardinality constraint, $ \sum_{i=1}^n x_i \leq k, $ and $d(\cdot,\cdot)$  is a metric,   the problem is also coined $\MSD_k$~\cite{borodin2012max}. 
Constant factor approximation algorithms for $\MSD_k$ have been developed in~\cite{ravi1994heuristic,hassin1997approximation}. Birnbaum and Goldman~\cite{birnbaum2009improved} presented an algorithm with approximation factor converging to $1/2$.  This is tight under the assumption that the \emph{planted clique} problem~\cite{alon2011inapproximability} is hard, see~\cite{borodin2012max}.  Baur and Fekete~\cite{baur2001approximation} have shown that this problem has a PTAS for $X \subset \R^d$ and $d(\cdot,\cdot)$ being the $\ell_1$-distance, provided that the dimension $d$ is \emph{fixed}. Bhattacharya et al.~\cite{bhattacharya2011consideration} developed a $1/2$-approximation algorithm for $\MSD_k$ where the objective function is replaced by $x^TDx + c^Tx$ for some $c \in \R^n$. This has been useful in accommodating also \emph{scores of documents} in the  objective function. 

Recently, Abbassi et al.~\cite{abbassi2013diversity} have shown that $\MSD$ has a $1/2$-approximation algorithm if $d(\cdot,\cdot)$  is a metric and $Ax \leq b$ models the independent sets of a \emph{matroid}. 
 This is particularly relevant in situations where  documents are partitioned into subsets $D_1,\dots,D_\ell$ and only $p_i$ results should be returned from partition $D_i$ for each $i$. The possible sets are then \emph{independent sets} of a \emph{partition matroid}. The case of one cardinality constraint only is subsumed by $\ell=1$. Thus, the tightness of their result also follows from the planted clique assumption as described in~\cite{borodin2012max}.

\subsubsection*{Contributions of this paper}
\label{sec:main-contr-this}

The $\frac{1}{2}+\epsilon$ 
hardness of $\MSD_k$~\cite{borodin2012max}
is based on a metric that does not play a prominent role as a
similarity measure. Are there better approximation algorithms,
possibly polynomial time approximation schemes, for other relevant
distance metrics?

We give a positive answer to this question for
the case where $d(\cdot,\cdot)$
is a distance of \emph{negative type}. We review the notion of
negative-type distances in Section~\ref{sec:theory-embeddability} and
here only note that the previously mentioned distances, like the ones
stemming from the $\ell_2$
and $\ell_1$-norm,
as well as the \emph{cosine-distance} and the \emph{Jaccard distance}
are, among many other relevant distance functions, of negative
type. Our main result is the following. 
\begin{theorem}
\label{thr:5}  
  There exists a polynomial time approximation scheme for $\MSD$
  for the case that $d(\cdot,\cdot)$
  is of negative type and $Ax \leq b$
  is a matroid constraint. In particular, there is a polynomial time
  approximation scheme for  $\MSD_k$ for negative type distances.
\end{theorem}
 \noindent 
% This is a significant improvement over the previous $2$-approximation algorithms for general metrics~\cite{birnbaum2009improved,abbassi2013diversity} and the PTAS for the $\ell_1$-norm in fixed dimension~\cite{baur2001approximation}. 
Theorem~\ref{thr:5}   is shown  by  following these two steps.

\begin{enumerate}[a)] 
\item We show that one can compute a fractional solution $x^* \in \R^{n}$ that fulfills all the constraints $Ax \leq b$ and satisfies ${x^*}^TDx^* \geq OPT$, where $OPT$ is the objective-function value of the optimal solution.
More precisely, our algorithm does not optimize over the natural relaxation,
but only over a family of slices of it, one for each possible
$\ell_1$-norm of the solution vector.
The key property we show and exploit is that the optimization problem on
each slice is a convex optimization problem that can be attacked by
standard techniques, despite the fact that the natural relaxation is not convex.
This allows us to  obtain an optimal solution to the relaxation 
via the ellipsoid method for a wide family of constraints $Ax\leq b$,
even if only a separation oracle is given; in particular, this includes
matroid polytopes~\cite{GroetschelLovaszSchrijver88}.  \label{item:4}

%The algorithm runs in polynomial time if the \emph{separation problem}  for the inequalities  $Ax \leq b$ can be solved  in polynomial time, which is the case for the matroid polytope~\cite{GroetschelLovaszSchrijver88}. 
\item 
For the case in which $Ax \leq b$ describes the convex hull of independent sets of a matroid, 
we furthermore describe a polynomial-time \emph{rounding algorithm} that computes an integral feasible solution $\bar{x}$ which satisfies 
$
    \bar{x}^TD\bar{x} \geq \left(1 - c \cdot \frac{\log k}{k}\right){x^*}^TD{x^*}, 
$
for some universal constant $c$ and where $k$ is the \emph{rank} of the matroid. 
\end{enumerate}

\noindent 
% We note that 
% this fractional solution $x^*$ is, in general, not the optimal solution of the relaxation that one obtains by replacing $x_i \in \{0,1\}$  by $0\leq x_i \leq 1$ for $1\leq i \leq n$.
% \todo{RZ: The previous sentence may be a bit misleading if the reader thinks about the matroid polytope for which we indeed solve the natural relaxation.}
% In fact, we do not know whether such an optimal solution of the relaxation can be computed in polynomial time. 
%Still, such an $x^*$ can be efficiently computed by convex optimization techniques via a suitable convexification.
%
Thus, step~\ref{item:4}) is via a suitable convexification 
of a non-convex relaxation. Convexifications  
have proved useful in the design of   approximation algorithms before, see for example \cite{skutella2001convex}. 

We also want to mention that we obtain similar results for the case where  the objective function is a combination of max-sum dispersion and linear scores, a scenario that has been considered in~\cite{bhattacharya2011consideration}.  
Finally, to complement our results, we prove $\classNP$-hardness of $\MSD_k$ for negative type distances. We prove as well that, for the rounding algorithm mentioned in point b), the approximation factor of $1 - O\left(\frac{\log k}{k}\right)$ almost matches the integrality gap of our relaxation, which we can show to be at least $1-\frac{1}{k}$.

\section{Preliminaries}
\label{sec:preliminaries-1}
In this section, we review some preliminaries that are required for the understanding of this paper. A \emph{polynomial time approximation scheme (PTAS)} for an optimization problem in which the objective function is maximized, is an algorithm that, given an instance and an $\epsilon >0$, computes a solution that has an objective function value of at least  $(1-\epsilon) \cdot OPT$, where $OPT$ denotes the the objective function value of the optimal solution.  Clearly, our rounding algorithm  
 is a PTAS since we can compute the optimal solution in the case where $\epsilon < c \cdot \frac{\log k}{k}$ by brute force. 

\subsubsection*{Norms and embeddings }  
\label{sec:theory-embeddability}

Our results rely heavily on the theory of embeddings. We review some
notions  that are relevant for us and refer to
\cite{DezaLaurent97,matouvsekEmbeddings} for a thorough account. For a
vector $v=(v^1,\cdots,v^t)^T\in\R^t$,
we define the $l_p$
norm in the usual way, $\|v\|_\infty:=\max_{1\leq i \leq t} |v^i|$,
and $\|v\|_p:=(\sum_{i=1}^t |v^i|^p)^{1/p}$
for $p\geq 1$,
and we extend this last definition to $0< p< 1$,
even if they are not proper norms as they do not respect the triangle
inequality. For $0< p\leq \infty$,
the space $(X,d)$
is $l_p$-\emph{embeddable}
if there is a dimension $t$
and a function $v:X\rightarrow\R^t$ 
(the isometric embedding),
such that for all $i,j\in X$
we have $d(i,j)=\|v_j-v_i\|_p$.
Any finite metric
space is $l_\infty$-embeddable
with the \emph{Fr\'echet embedding}~\cite{frechet1910dimensions}
$v_i^j=d(i,j)$ for $i,j\in X$.
 %\footnote{We stress the fact that the corresponding embedding $v$ and dimension $t$ need not be known. For instance, finding an embedding for an $l_2$-embeddable space is polynomial-time solvable, while the respective task for an $l_1$-embeddable space is NP-complete \cite{avis91}. If $([n],d)$ is $l_p$-embeddable, then the minimum necessary dimension $t$ is at most $n-1$ for $p=2$ and $p=\infty$, and at most $\binom{n}{2}$ for all $p\geq 1$ \cite{ball,fichet}.}\\

% Given a distance $d(\cdot,\cdot)$ and a function $f:\R_{\geq 0}\rightarrow \R_{\geq 0}$ with $f(0)=0$, we represent by $f(d)$ the distance defined by $f(d)(i,j)=f(d(i,j))$. Naturally, the property of $l_2$-embeddability will be of central interest to us.

%However, it will be more convenient to work with the concept of \emph{negative %type}: 

%\noindent 
For the remainder of this paper, we assume  that $X = \{1,\dots,n\}$ and $n \geq 2$. Let $b_1,\dots,b_n$ be real coefficients. The inequality 
\begin{align}\label{def:NT}
  \sum_{1 \leq i,j \leq n} b_ib_j x_{ij} \leq 0
\end{align}
with variables $x_{ij}$ is a \emph{negative type} inequality if $\sum_{i=1}^n b_{i} = 0$. The distance space $(X,d)$ is of negative type if $d(\cdot,\cdot)$ satisfies all negative type inequalities, i.e., $\sum_{1 \leq i , j \leq n}  b_ib_j d(i,j) \leq 0$ holds for all $b_1,\dots, b_n \in \R$ with $\sum_{i=1}^n b_i = 0$. Schoenberg~\cite{schoenberg1938metricI,schoenberg1938metricII} characterized the metric spaces that are $\ell_2$-embeddable as those, whose square distance is of negative type. 
 
\begin{theorem}[\cite{schoenberg1938metricI,schoenberg1938metricII}]
\label{thr:1}
  A finite  distance space $(X,d)$ is of negative type if and only if $(X,\sqrt{d})$ is $\ell_2$-embeddable. 
\end{theorem}
\noindent 
The following assertions, which help identifying distance spaces of negative type, can be found in~\cite{DezaLaurent97}.  

%\begin{theorem}\label{embedding} Let $(X,d)$ be a finite distance space.
\begin{enumerate}[i)]
\setlength\itemsep{0em}
%\item If $(X,d)$ is $l_2$-embeddable, then it is $l_p$-embeddable for all $1\leq p\leq \infty$. \cite{dor1976potentials} 
\item If $(X,d)$ is a metric space, then $\left(X,d^{\log_2\left(\frac{n}{n-1}\right)}\right)$ is of negative type. \cite{deza1990metric} \label{item:3}
\item For any $0< \alpha\leq p\leq 2,$ if $(X,d)$ is $l_p$-embeddable, then $(X,d^\alpha)$ is of negative type. \cite{schoenberg1938metricII} \label{item:1}
\item If $(X,d)$ is of negative type, then $(X,f(d))$ is also of negative type for any of the following functions: $f(x)=\frac{x}{1+x}$, $f(x) = \ln(1+x)$, $f(x) = 1-e^{-\lambda x}$ for $\lambda>0$, and $f(x) = x^\alpha$ for $0\leq \alpha \leq 1$. \cite{schoenberg1938metricI} \label{item:2}
\end{enumerate}

We  now list some distance functions that are of negative type and which are often used in information retrieval and web search. 
The $\ell_1$-metric is of negative type. This follows from the assertion~\ref{item:1})  above with $\alpha = p = 1$.  
In fact, any $\ell_p$-metric with $1 \leq p \leq 2$ is of negative type. 
The $\ell_1$-metric is a prominent similarity measure in information retrieval~\cite{manning2008introduction} in particular when using sketching techniques~\cite{lv2004image} where data points are represented by small-dimensional bit-vectors whose Hamming-distance approximates the distance of the corresponding points. 
%\item 
Also the \emph{cosine distance}, which measures the distance of two points on the sphere $S^{(t-1)}$ by the angle that they enclose, is of negative type. This follows from a result of Blumenthal~\cite{blumenthal1953theory}, see also~\cite{DezaLaurent97}.  
%\item 
For subsets $A,B$ of a finite ground set $U$, the \emph{Jaccard distance} $d(A,B)=1 - \frac{|A \cap B| }{ |A \cup B|} $ is of negative type~\cite{gower1986metric}. Similarly, many other distances of sets are of negative type, such as \emph{Simple Matching} $\frac{|A\bigtriangleup B|}{|U|}$, \emph{Russell and Rao} $1-\frac{|A \cap B|}{|U|}$, and \emph{Dice} $\frac{|A \bigtriangleup B|}{|A| + |B|}$ distances (see~\cite[Table 5.1]{Pekalska:2005:DRP:1197035} for a more complete list). Assertion~\ref{item:2}) above presents some examples of transformations of distance spaces that preserve this property, and thus permit to construct new spaces of negative type from existing ones.

It is important to remark that distance spaces of negative type are in general not metric, or vice-versa. Hence results for these two families of distance spaces are not directly comparable. For instance, the Dice distance mentioned before is not metric. 
%\end{enumerate}

\subsubsection*{Matroids and the matroid polytope}
We mention some basic definitions and results on matroid theory, see, e.g.,~\cite[Volume B]{Schrijver03} for a thorough account. 
A matroid $\M$
over a finite  ground set $X$
is a tuple $\M = (X,\I)$,
where $\I \subseteq 2^X$
is a family of \emph{independent} sets with the following properties. 
\begin{itemize}
\setlength\itemsep{0em}
\item[(M1)]  $\emptyset\in\I$. 
\item[(M2)] If $A\subseteq B$ and $B \in \I$, then  $ A\in \I$. 
\item[(M3)] If $A,B\in \I$ and  $|A|>|B|$, then there exists an element  $e\in A\setminus B$ such that $B\cup \{e\} \in \I$.  
\end{itemize}

In particular, the family of all
subsets of $X$ of cardinality at most $k$, i.e., $\I = \{ S\subseteq X \colon \mid S \mid \leq k\}$, forms a matroid known as \emph{uniform matroid of rank $k$}, often denoted by $U_n^k$.   Thus, $\MSD_k$ can be understood as picking an independent set $S \in \I$ from the uniform matroid that maximizes the sum of the pairwise distances.

The \emph{rank} $r(A)$
of $A \subseteq X$
is the maximum cardinality of an independent set contained in $A$.
Any inclusion-wise maximal independent set $B$
is called a \emph{basis}, and a direct consequence of the definition
of matroids is that all bases have the same cardinality $r(X)$,
called the \emph{rank of the matroid}.

%For any set $A\subseteq X$, define its characteristic vector $\chi^A\in \R^n$ as $\chi^A_i=1$ if $i\in A$, and $0$ otherwise. 
%And for a vector $x\in \R^n$ and a set $A\subseteq X$. %, define the restricted %vector $x^A$ as $x^A_i=x_i \cdot \chi^A_i$. 
The \emph{matroid polytope} $P(\M)$ is the convex hull of the characteristic vectors of the independents sets of the matroid $\M$. 
%The matroid base polytope 
%is the convex hull of the characteristic vector of all the bases of
%$\M$,
%and can also be characterized by the constraints imposed by the set
%ranks; i.e.
It can be described by the following inequalities $P(\M)=\{x\in \R^n_{\geq 0} \colon \sum_{i \in A} x_i \leq r(A) \ \forall
A\subseteq X \}$. The base polytope of a matroid $\M$ of rank $k$ is the convex hull of all characteristic vectors of bases of $\M$, and is given by $P(\M)\cap \{x\in \R^n \colon \sum_{i=1}^n x_i = k\}$.

\subsubsection*{Convex quadratic programming}
\label{sec:conv-quadr-progr}

A \emph{quadratic program} is an optimization problem of the form 
\begin{displaymath}
  \min \{ x^T Q x + c^T x \colon x \in \R^n, \, Ax \leq b\}, 
\end{displaymath}
where $Q \in \R^{n \times n}$, $c \in \R^n$, $A \in \R^{m\times n}$ and $b \in \R^m$. If $Q$ is positive semidefinite, then it is a \emph{convex quadratic program}. Convex quadratic programs can be solved in polynomial time with the ellipsoid method~\cite{Khachiyan79}, see, e.g.,~\cite{kozlov1980polynomial}.
This also holds if $Ax \leq b$ is not explicitly given but the \emph{separation problem} for $Ax \leq b$ can be solved in polynomial time~\cite{GroetschelLovaszSchrijver88}.
In this case the running time is polynomial in the input size of $D$ and the largest
binary encoding length of the coefficients in $A$ and numbers in $b$.
The separation problem for the matroid polytope $P(\M)$ can be solved in polynomial time, provided that one can efficiently decide whether a set $S \subseteq X$ is an independent set, 
see~\cite{GroetschelLovaszSchrijver88}. Moreover,
the largest encoding length of the numbers in the above-mentioned
description of the matroid polytope is $O(\log n)$.
Thus a convex quadratic  program over the matroid polytope can be
solved in polynomial time.

\section{A  relaxation that can be solved by convex  programming}
\label{sec:relax-solv-conv} 

We now describe how to efficiently compute a fractional point $x^*$ for the relaxation of \eqref{eq:1} with an objective value ${x^*}^TDx^* \geq OPT$. 
The function $f: \R^n \rightarrow \R$, $f(x) = x^T D x$ is in general non-concave, even if the distances $d(i,j)$ are  $\ell_2$-embeddable or of negative type. % The following lemma  allows us to compute a fractional point $x^*$  in the relaxation with convex quadratic programming which is, although not necessarily the optimal solution of the relaxation, satisfying ${x^*}^TDx^* \geq OPT$ which makes it suitable for rounding algorithms.  
However, we have the following useful observation.

\begin{lemma}
  \label{lem:sliceConvex}
  Let $(X,d)$ be a finite distance space of negative type, then 
  \begin{displaymath}
    f(x) =  x^T D x
  \end{displaymath}
  is a concave function over the domain  $ \{ x \in \R^n \colon \sum_{i=1}^n x_i = \alpha\}$  for each fixed $\alpha \in \R$.  
\end{lemma}

\begin{proof}
  By  Theorem~\ref{thr:1}, the distance $\sqrt{d}$ is $\ell_2$-embeddable; let  $v_1,\dots v_n \in \R^t$ be a corresponding embedding, i.e., 
  \begin{eqnarray*}
    d(i,j) & = &  \norm{v_i - v_j}_2^2 \\
           & = &  (v_i - v_j)^T (v_i - v_j) \\
           & = & \norm{v_i}_2^2 - 2 v_i^Tv_j + \norm{v_j}_2^2. 
  \end{eqnarray*}
For $x \in \R^n$ with $\sum_{i=1}^n x_i = \alpha$,  $ x^TDx$ can be written as 
\begin{eqnarray*}
  x^TDx & = &  \sum_{1\leq i,j \leq n} x_ix_j \left( \norm{v_i}_2^2 - 2 v_i^Tv_j + \norm{v_j}_2^2 \right)\\
        & = &   \sum_{j=1}^n x_j \sum_{i=1}^n (x_i \norm{v_i}_2^2 )
          - 2 \sum_{1\leq i,j\leq n} x_i (v_i^T v_j) x_j
          + \sum_{i=1}^n x_i \sum_{j=1}^n (x_j \norm{v_j}_2^2 )  \\
       & = & 2 \alpha c^T x - 2 \cdot x^T Q x,
\end{eqnarray*}
where $Q \in \R^{n\times n}$  is the positive semidefinite matrix 
with $Q_{ij} = v_i^Tv_j$, hence $x^TQx$ is convex, and $c \in \R^n$ is the vector $c^T = (\norm{v_1}_2^2,\dots,\norm{v_n}_2^2)$. Thus $x^T D x$ is concave on the domain $\{x\in \R^n \colon \sum_{i=1}^n x_i = \alpha\}$.
\end{proof}

\begin{remark}
Notice that the matrix $Q$ and the  vector $c$ can easily be derived
from the distances. For this observe that if
$v_1,\dots,v_n \in \R^t$ is an embedding of $\sqrt{d}$,
then $v_1-u, v_2-u,\dots,v_n-u$ is another such embedding,
for each $u \in \R^t$.
Hence, we can assume $v_1 = \mathbf{0}$, which implies
$c_i=\norm{v_i}_2^2 = d(1,i)$ and thus
$Q_{ij} = \frac{1}{2} \left(d(1,i) + d(1,j) - d(i,j)\right)$.
Hence, the embedding does not need to be known to describe $Q$ and $c$. 
\end{remark}

Using Lemma~\ref{lem:sliceConvex}, we can efficiently determine
a relaxed solution for $\MSD$ on distance spaces of negative type
for a wide class of constraints,
by solving a family of convex problems, one for each possible
$\ell_1$ norm of the solution vector.
There is a rich set of algorithms for convex optimization
problems as we encounter here. In the following theorem, for simplicity,
we focus on consequences stemming from the ellipsoid algorithm.
The ellipsoid algorithm has the advantage that it only needs a separation
oracle, and often allows us to obtain an optimal solution without
any error. As a technical requirement, we need that the coefficients
of the underlying linear constraints have small encoding lengths,
which holds for most natural constraints.

\begin{theorem}
  \label{thr:2}
  Consider the max-sum dispersion problem with general linear constraints, for which the separation problem can be solved in polynomial time

  \begin{equation}
    \label{eq:2}
    \begin{array}{ll}
      \text{maximize} &   x^T D x\\
      \\
      \text{subject to} & Ax \leq b \\
      \\
                      & x_i \in \{0,1\} \, \text{ for }\,i \in X. 
    \end{array}
  \end{equation}

  If $d(\cdot,\cdot)$ is of negative type, then one can compute a fractional point $x^* \in [0,1]^n$ satisfying $Ax \leq b$ with ${x^*}^T D x^* \geq OPT$ in time polynomial in the input and the maximal binary encoding length of any coefficient or right-hand side of $Ax \leq b$. 
\end{theorem}

Before proving the theorem, we briefly discuss the above-mentioned dependence
of the running time on the encoding length. Notice
that if $A x \leq b$ is given explicitly,
then the claimed point $x^*$ is always obtained in polynomial time because the
encoding length of $A x \leq b$ is part of the input. Similarly,
Theorem~\ref{thr:2} implies that we can obtain $x^*$ efficiently for matroid polytopes,
since the inequality-description mentioned in Section~\ref{sec:preliminaries-1}
has only $\{0,1\}$-coefficients
and right-hand sides within $\{1,\dots, n\}$.
We highlight that our techniques can often be used even if the encoding length
condition is not fulfilled by accepting a small additive error.

\begin{proof}
  Consider the constraints $\sum_{i=1}^n x_i = \alpha$
  for $\alpha \in \{1,\dots,n\}$. Using the notation of Lemma~\ref{lem:sliceConvex}, we can re-write the objective function as 
$2 \alpha c^T x - 2 \cdot x^T Q x$ if the constraint $\sum_{i=1}^n x_i = \alpha$ is added to the constraints $Ax \leq b$. 
Using the ellipsoid method, we 
solve each of the following  $n$
convex quadratic programming problems  that are parameterized by $\alpha \in \{1,\dots,n\}$ 
   \begin{equation}
    \label{eq:3}
    \begin{array}{ll}
      \text{minimize} &     2 \cdot x^T Q x - 2 \alpha c^T x\\
      \\
      \text{subject to} & Ax \leq b \\
      \\
                        & \sum_{i=1}^n x_i = \alpha\\
                       \\
                      & 0 \leq x_i \leq 1, \text{ for }\,i \in X, 
    \end{array}
  \end{equation}
with optimum solutions $x^{*,1},\dots,x^{*,n}$ respectively;
see~\cite{kozlov1980polynomial,GroetschelLovaszSchrijver88} for
details on why an optimal solution can be obtained (without any additive error
which is typical for many convex optimization techniques). 
Since each feasible $\bar{x} \in \{0,1\}^n$
satisfies one of these constraints, $x^*$ being one of these solutions with largest objective function value is a point satisfying the claim. Clearly, $x^*$ can be computed in polynomial time. 
\end{proof}

\begin{remark}
We notice that for very simple constraints, like a cardinality constraint
where at most $k$ elements can be picked,
a randomized rounding approach \cite{raghavan1987randomized} can  now  be  employed. More precisely,
when working with a cardinality constraint one can simply scale
down the fractional solution $x^*$ satisfying $\sum_{i=1}^n x_i^*=k$
by a $(1-\epsilon)$-factor to obtain $y=(1-\epsilon) x^*$,
and then round each component of $y$ independently.
Independent rounding preserves the objective
in expectation, and the number of picked elements is $(1-\epsilon) k$
in expectation and sharply concentrates around this value due
to Chernoff-type concentration bounds. However, this simple rounding
approach fails for more interesting constraint families like
matroid constraints.
\end{remark}

\section{Negative type MSD under matroid constraint}
Consider the MSD problem for the case that the distance space is of negative type and $Ax\leq b$ is a matroid constraint. The main result of this section is the following theorem which immediately implies Theorem~\ref{thr:5}. 
\begin{theorem}
  \label{thr:4}
  There exists a deterministic algorithm for the $\MSD$
  problem in distance spaces of negative type with a matroid
  constraint, which outputs in polynomial time a basis $B$
  with
  $\left({\chi^B}\right)^T D \chi^B \geq \left(1-c\frac{\log
      k}{k}\right)\OPT$, where $k$
  is the rank of the matroid and $c$ is an absolute constant.
\end{theorem}

We solve the relaxation to obtain an optimal fractional point over the matroid polytope, as in Theorem \ref{thr:2}, and perform a deterministic rounding algorithm.
The suggested rounding procedure has similarities with pipage
rounding for matroid polytopes (see \cite{calinescu_2011_maximizing},
which is based on work in~\cite{ageev_2004_pipage})
and swap rounding~\cite{chekuri_2010_dependent}, in the sense that
it iteratively changes at most two components of the fractional
point until an integral point is obtained. However contrary to
these previous procedures we need to 
judiciously choose the two coordinates.
Also our analysis differs substantially from the above-mentioned
prior rounding procedures on matroids, since we deal with a quadratic
objective function were we must accept a certain loss in the
objective value due to rounding, because there is a strictly
positive integrality gap.
(Pipage rounding and swap rounding are typically applied in settings
where the objective function is preserved in expectation.)
Makarychev, Schudy, and Sviridenko~\cite{makarychev_2015_concentration}
build up on the swap rounding procedure and show how to obtain
concentration bounds for polynomial objective functions.
Their concentration results apply to general polynomial objective
functions with coefficients in $[0,1]$; however,
they are not strong enough for our purposes.

Our deterministic rounding algorithm exploits the fact that we are dealing with negative type distance spaces, and shows that only a very small loss in the objective value is necessary to obtain an integral solution. In order to bound this loss we use a very general inequality, stemming from the definition of distance spaces of negative type, that compares the (fractional) dispersion of two sets to that of its union. Given a vector $x\in \R^n$ and a set $S\subseteq\{1,\cdots,n\}$, we define the restricted vector $x^S$ as $x^S_i=x_i$ if $i\in S$, and 0 otherwise.

\begin{lemma}\label{lemmaIneq}
Let $D\in\R^{n\times n}$ be the matrix representing a negative type distance space. Given a vector $x\in \R^n_{\geq 0}$ of coefficients, and two disjoint sets $A,B\subseteq [n]$ such that $\|x^A\|_1>0$ and $\|x^B\|_1>0$, we have 
\begin{displaymath}
\frac{\left({x^{A\cup B}}\right)^T D x^{A\cup B}}{\|x^{A\cup B}\|_1}
\geq \frac{\left({x^A}\right)^T D x^A}{\|x^A\|_1}
+ \frac{\left({x^B}\right)^T D x^B}{\|x^B\|_1}.
\end{displaymath}
\end{lemma}

\begin{proof}
Define the vector $b\in \R^n$ as $b=\frac{x^A}{\|x^A\|_1} - \frac{x^B}{\|x^B\|_1}$. Since $\sum_{i=1}^n b_i = 0$, the inequality $b^T D b\leq 0$ is of negative type. Expanding it yields

$$ 0\geq  \left(\frac{x^A}{\|x^A\|_1} - \frac{x^B}{\|x^B\|_1}\right)^T D \left(\frac{x^A}{\|x^A\|_1} - \frac{x^B}{\|x^B\|_1}\right)= \frac{\left({x^A}\right)^T D x^A}{\|x^A\|_1^2} + \frac{\left({x^B}\right)^T D x^B}{\|x^B\|_1^2} - 2\frac{ \left({x^A}\right)^T D x^B}{\|x^A\|_1 \|x^B\|_1}.$$
Hence $2 \left({x^A}\right)^T D x^B \geq \frac{\|x^B\|_1}{\|x^A\|_1}
\left({x^A}\right)^T D x^A + \frac{\|x^A\|_1}{\|x^B\|_1} \left({x^B}\right)^T D x^B.$ Finally,
\begin{align*}
{\left(x^{A\cup B}\right)}^T  D x^{A\cup B} = & \left(x^A + x^B\right)^T D (x^A + x^B) = \left({x^A}\right)^T D x^A + \left({x^B}\right)^T D x^B + 2 \left({x^A}\right)^T D x^B\\
 \geq & \left(\|x^A\|_1 + \|x^B\|_1\right) \left( \frac{\left({x^A}\right)^T D x^A}{\|x^A\|_1} + \frac{\left({x^B}\right)^T D x^B}{\|x^B\|_1}\right)\\
 = & \|x^{A\cup B}\|_1 \left( \frac{\left({x^A}\right)^T D x^A}{\|x^A\|_1} + \frac{\left({x^B}\right)^T D x^B}{\|x^B\|_1}\right).    
\end{align*} 
\end{proof}

Consider an MSD instance consisting of a distance space of negative type represented by a matrix $D\in\R^{n\times n}$, and a matroid $\M$ over the ground set $X=\{1,\ldots,n\}$ of rank $k$. We assume that one can efficiently decide whether a set $S\subseteq X$ is independent. So we apply Theorem \ref{thr:2} to find a fractional vector $x^*$ over the matroid polytope $P(\M)=\left\{x\in \R^n_{\geq 0} \colon \sum_{i \in A} x_i \leq r(A) \ \forall A\subseteq X \right\}$, with $\left({x^*}\right)^T D x^* \geq \OPT$. Due to the monotonicity of the diversity function $x^T D x$, we can assume that $\|x^*\|_1=k$, i.e.,
$x^*$ is on the base polytope of the matroid $\M$.\footnote{Indeed, using standard techniques from matroid optimization (see~\cite[Volume B]{Schrijver03}), one can, for any point $y\in P(\M)$, determine a point $z\in P(\M)\cap \{x\in \R^n \colon \|x\|_1=k\}$ satisfying $z\geq y$ component-wise. Hence, if the fractional point $x^*$ we obtain is not on the base polytope, we can replace it efficiently with a point on the base polytope that dominates it and therefore has no worse objective value than $x^*$ due to monotonicity of the considered objective.}
We describe now a deterministic rounding algorithm that takes $x^*$ as input, and outputs in polynomial time a basis $B$ of $\M$ with $\left({\chi^B}\right)^T D \chi^B\geq \left(1-O\left(\frac{\log k}{k}\right)\right)\left({x^*}\right)^T D x^*\geq \left(1-O\left(\frac{\log k}{k}\right)\right)\OPT$.

In the remainder of the section, for any vector $x \in P(\M)$ we ignore the elements $i$ with $x_i=0$ and assume without loss of generality that $x$ has no zero components. We call an element $i\in X$ \emph{integral} or \emph{fractional} (with respect to $x$), respectively, if $x_i=1$ or $x_i<1$, and we call a set $S\subseteq X$ \emph{tight} or \emph{loose}, respectively, if $\|x^S\|_1=r(S)$ or $\|x^S\|_1<r(S)$. We will need the following result about faces of the matroid polytope, which is a well-known consequence of combinatorial uncrossing (see~\cite{giles_1975_submodular}, or~\cite[Section 44.6c in Volume B]{Schrijver03}).

\begin{lemma}\label{lemmaChain}
Let $x\in P(\M)$ with $x_i\neq 0$ for $i \in \{1,\dots, n\}$, and let
$\emptyset = S_0\subsetneq S_1 \subsetneq \cdots \subsetneq S_p=X$
be a (inclusion-wise) maximal chain of tight sets with respect to $x$,
i.e., $\sum_{i\in S_l} x_i = r(S_l)$ for $l\in \{1,\dots, p\}$.
Then the polytope $P(\M)\cap \{y\in \R^n_{\geq 0} \colon \sum_{i\in S_l} x_i = r(S_l) \text{ for } l\in \{1,\dots, p\}\}$ defines the minimal face of
$P(\M)$ that contains $x$.
(In other words, all other $x$-tight sets are implied by the ones in the chain.)
\end{lemma}

Also, given a point $x\in P(\M)$, one can efficiently find a maximal
chain of tight sets as described in Lemma~\ref{lemmaChain}.
Our algorithm starts with such
a chain $\emptyset = S_0\subsetneq S_1 \subsetneq \cdots \subsetneq S_p=X$ for the vector $x^*$. For $1\leq l\leq p$ define the set $R_l=S_l \setminus S_{l-1}$; we call these sets \emph{rings}. The rings form a partition of $X$, their weights $\|\left({x^*}\right)^{R_l}\|_1=r(S_l)-r(S_{l-1})$ are strictly positive integers whose sum is $k$, and each ring $R_l$ either consists of a single integral element, or of at least 2 elements, all fractional. This is because whenever $i\in R_l$ is integral, the set $S_{l-1}\cup\{i\}$ is tight, so it can be added to the chain. We call the rings integral or fractional, accordingly. We start with $x=x^*$,
we iteratively change two coordinates of $x^*$ without
leaving the minimal face of the matroid polytope
on which $x^*$ lies; one coordinate will be increased
and the other one decreased by the same amount.

%When changing two coordinates of $x^*$, we do it in such a way
%that we hit a new tight constraint, either a non-negativity constraint
%or a matroid constraint defined by a new tight set $S$ not
%implied by the current chain. Using combinatorial uncrossing
%a new tight set $S$ can easily be transformed into a tight
%set that can be added to the current chain, leading to a chain
%that is by one set larger
%(see~\cite[Section 44.6c in Volume B]{Schrijver03}).
%%
%We will use this fact to maintain a efficiently
%maintain a chain of tight sets. 

\subsection*{The rounding procedure}

The rounding of $x^*$
proceeds in iterations, and stops when all elements are integral. 
 Among all pairs of fractional elements within the
same ring, select the pair $i,j$
that minimizes the term $x_i x_j d(i,j)$.
We perturb vector $x$ by adding to $x_i$
and subtracting from $x_j$
a certain quantity $\epsilon$.
The dispersion $x^TDx$ is linear in $\epsilon$
except for the term $2x_i x_j d(i,j)$,
hence we can select the sign of $\epsilon$
so that the value of $x^T Dx-2x_i x_j d(i,j)$
does not decrease. We assume without loss of generality that 
this choice is $\epsilon>0$, so $x_i$ is increasing and $x_j$ decreasing.
We increment $\epsilon$ until a new tight constraint appears. 
If the constraint corresponds to $x_j$ becoming zero, 
we erase that element and end the iteration step. 
Otherwise, a previously loose set $S\subseteq X$ becomes tight, 
and $S$ must contain $i$ but not $j$, 
else its weight $\|x^S\|_1$ would not increase during this process. 
If the ring containing $i$ and $j$ is $R_l=S_l\setminus S_{l-1}$, 
then the set $S'=(S\cup S_{l-1}) \cap S_l$ is also 
tight,\footnote{This follows from the uncrossing property: 
if $A$ and $B$ are tight sets then $A\cup B$ and $A\cap B$ are also tight. 
This property is a consequence of the submodularity of the matroid rank function.} 
and it also contains $i$ but not $j$, 
hence $S_{l_1}\subsetneq S' \subsetneq S_l$ (see Figure \ref{fig:1}).
We add $S'$ to the chain, update the list of rings, 
and end the iteration step.

\begin{figure}
  \centering
  \subfigure[The fractional ring containing the pair $i,j$ with minimal value $x^m_ix^m_id(i,j)$.]{\includegraphics[height=2.9cm]{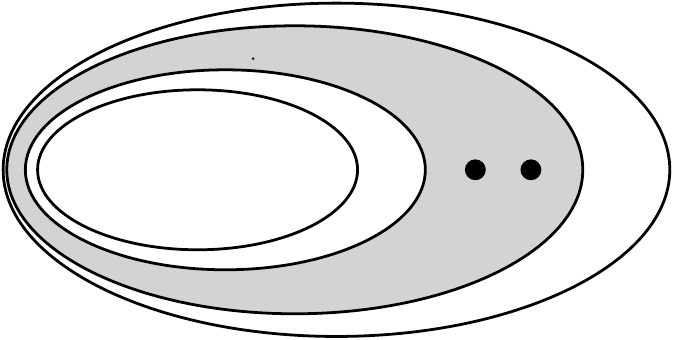}} \quad \quad \quad
  \subfigure[Elements $i$ and $j$ are separated by a new tight set that fits in the chain structure.]{\includegraphics[height=2.9cm]{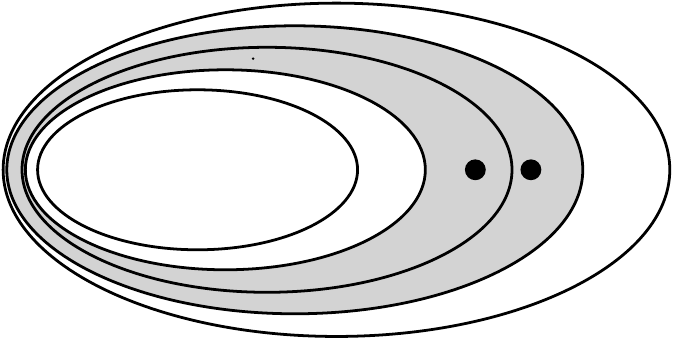}}
  \caption{The refinement of a fractional ring in an iteration of the rounding procedure. }
  \label{fig:1}
\end{figure}

We now argue that the above-described rounding procedure runs in
polynomial time and computes the characteristic vector $\chi^B$
of a basis $B$ of the matroid $\M$ with
  \begin{equation}
    \label{eq:4}
    \left({\chi^B}\right)^T D \chi^B \geq \left(1-c\frac{\log k}{k}\right) {x^*}^TDx^*     
  \end{equation}
where $k$ is the rank of the matroid and $c$ is an absolute constant.

\medskip 
\noindent 
At any stage of the algorithm, if $q$ is the number of fractional rings and $f$ is the number of fractional elements, the number of iterations remaining is at most $f-q$. This is because the value $f-q$ can never be negative, and it decreases in each iteration.  Either $f$ decreases, or $q$ increases, or $q$ decreases by 1 but $f$ decreases by at least 2 (any disappearing fractional ring has at least 2 fractional elements that become integral).  

Suppose there are $M$ iterations. We enumerate them in reverse order, and add a superscript $m$ to all variables to signify their value when there are $m$ iterations remaining. Hence $x^0= \chi^B$ is the integral output vector, $x^1$ is the vector at the beginning of the last iteration, and so on until $x^M=x^*$. Clearly, all vectors $x^m$ are in $P(\M)$, and their weights $\|x^m\|_1=k$ remain unchanged, so each $x^m$ is on the base polytope, and $x^0$ will thus be a characteristic vector of a basis in $\M$. From the previous claim we know that $m\leq f^m - q^m$, and in particular $M\leq n$, so the algorithm runs in polynomial time. For $1\leq m\leq M$, define $\loss^m=\left({x^m}\right)^T D x^m - \left({x^{m-1}}\right)^T D x^{m-1}$, so the total additive loss incurred in the rounding algorithm is $\sum_{m=1}^M \loss^m$. We postpone for a moment the proof of the following inequality.  

\begin{lemma}
  \label{lem:1}
  The loss in iteration $m$ is bounded by 
  \begin{displaymath}
    \loss^m \leq \min\left\{ \frac{2}{m\cdot k}, \frac{2}{m^2}\right\} {x^*}^TDx^*. 
  \end{displaymath}
\end{lemma}
\noindent 
The total additive loss incurred by the algorithm is 
\begin{align*}
\left({x^*}\right)^T D x^* - \left({x^0}\right)^T D x^0 &=
  \sum_{m=1}^M \loss^m \leq \left({x^*}\right)^T D x^*\left(\sum_{m=1}^k \frac{2}{m\cdot k} +\sum_{m>k} \frac{2}{m^2}\ \right)\\
&\leq \left({x^*}\right)^T D x^*
  \cdot 2 \left(\frac{1+\ln k}{k} + \frac{1}{k}\right)
=
\left({x^*}\right)^T D x^* \cdot
  \left(\frac{4+2\ln k}{k}\right)\enspace,
\end{align*}
where the second inequality follows from
$\sum_{m=1}^k \frac{1}{m} \leq 1+\ln k$ and
$\sum_{m>k}\frac{1}{m^2} \leq \frac{1}{k}$.
In summary,
the algorithm finds a basis with dispersion
\begin{equation*}
\left({x^0}\right)^T D{x^0}\geq
\left({x^*}\right)^T D x^*\left(1-\frac{4+2\ln k}{k}\right)\geq \OPT\left(1-O\left(\frac{\log k}{k}\right)\right) 
\end{equation*}
which is our main result.

\begin{proof}[Proof of Lemma~\ref{lem:1}] 
 If the pair $i,j$ of fractional elements is chosen during the $m$-th iteration, then  $\loss^m \leq 2x^m_i x^m_j d(i,j)$. We 
bound this  term from above,  using the inequality in Lemma~\ref{lemmaIneq} multiple times over the vector $x^m$ and its partition into rings.  We skip the superscript $m$ to simplify notation and obtain 

$$\frac{x^T D x}{k} \geq \sum_{\text{ring } R} \frac{\left({x^R}\right)^T D x^R}{\|x^R\|_1}\geq \sum_{\text{fractional } R} \frac{\left({x^R}\right)^T D x^R}{\|x^R\|_1} \geq \sum_{\text{fractional } R} \frac{\binom{|R|}{2} \loss}{\|x^R\|_1}\enspace,$$
where the last inequality comes from $\loss^m \leq 2x^m_ix^m_j d(i,j)$, and the choice of elements $i$ and $j$ that minimizes this quantity over all pairs of elements within any fractional ring. We complete the above inequality in two ways. First, for every fractional ring $R$, $|R|\geq \|x^R\|_1$, so $\frac{\binom{|R|}{2}}{\|x^R\|_1}\geq \frac{|R|-1}{2}$ and thus
$$\frac{x^T D x}{k} \geq \frac{\loss}{2}\sum_{\text{frac. } R} (|R|-1)=\frac{\loss}{2}(f-q)\geq \frac{m}{2} \loss\enspace,$$
which implies
\begin{equation*}
\loss^m\leq \frac{2}{m\cdot k} \left({x^m}\right)^T D x^m \leq \frac{2}{m\cdot k} \left({x^*}\right)^T D x^*\enspace.
\end{equation*}

Next, $\binom{|R|}{2}\geq (|R|-1)^2/2$, and using a Cauchy-Schwarz inequality (for the third inequality below), we get:

\begin{align*}
x^T D x^T  \geq & \frac{k\cdot \loss}{2} \sum_{\text{frac. } R} \frac{(|R|-1)^2}{\|x^R\|_1} \geq \frac{\loss}{2} \left(\sum_{\text{frac. } R} \|x^R\|_1 \right)\left(\sum_{\text{frac. } R} \frac{(|R|-1)^2}{\|x^R\|_1}\right)\\
 \geq & \frac{\loss}{2} \left(\sum_{\text{frac.} R} (|R|-1) \right)^2 = \frac{\loss}{2} (f-q)^2\geq \frac{m^2}{2} \loss\enspace,
\end{align*}
and hence,
\begin{equation*}
\loss^m \leq \frac{2}{m^2} \left({x^m}\right)^T D x^m \leq \frac{2}{m^2} \left({x^*}\right)^T D x^*\enspace.
\end{equation*}

\end{proof}

\begin{remark}
The integrality gap of the convex program $\max\{x^T D x \ | \ x\in P(\M)\}$ above is at least $1-\frac{1}{k}$, which matches the approximation factor of our rounding algorithm up to a logarithmic term. Consider the matrix $D$ with $D_{i,j}=1$ for all $i\neq j$, which defines a distance space of negative type, and a cardinality constraint corresponding to the polytope $\left\{x\in [0,1]^n \ | \ \sum_i x_i = k\right\}$. An optimal solution is any $k$-set $B\subseteq X$, with value $\OPT=(\chi^B)^T D \chi^B=k(k-1)$; but the fractional vector $x^*=(k/n,\cdots,k/n)$ is feasible and has value $(x^*)^T D x^* = \frac{k^2}{n^2}n(n-1)$. Hence, $\frac{\OPT}{(x^*)^T D x^*}=\frac{k-1}{k}\frac{n}{n-1}\rightarrow 1-\frac{1}{k}$, as $n\rightarrow \infty$. 
\end{remark}

\begin{remark}
The previous approximation extends to the more general case of a combination of dispersion and linear scores, as follows. Consider the problem of maximizing the objective function $g(x)=x^T D x + w^T x$, with a matroid constraint of rank $k$, and where $D$ represents a distance space of negative type. The vector $w$ here corresponds to non-negative \emph{scores} on the elements of the ground set, and the objective is to find a feasible set with both high dispersion and high scores. The extra linear term does not change the concavity of $x^T D x$, so Lemma \ref{lem:sliceConvex} and Theorem \ref{thr:4} are valid for this problem and provide a fractional vector $x^*\in P(\M)$ with $g(x^*)\geq \OPT$. Moreover, $g(x)$ is still monotone, so we can assume that $\|x^*\|_1=k$. In each iteration of this section's rounding algorithm, $g(x)-2x_i x_j d(i,j)$ is linear in $\epsilon$, so we can bound the loss of value of $g(x)$ during this iteration by $2x_ix_jd(i,j)$, as before. Hence, the above analysis still holds and shows that the total loss is very small, even when comparing it only to the contribution of the quadratic term $(x^*)^T D x^*$ to the objective, thus ignoring the additional nonnegative term $w^T x^*$. Thus, we get the same approximation guarantee for this setting.
%Following a similar analysis as before, we can prove that $\loss^m\leq \frac{2}{m\cdot k} \left({x^m}\right)^T D x^m\leq \frac{2}{m\cdot k} g(x^m)\leq \frac{2}{m\cdot k} g(x^*)$, and $\loss^m\leq \frac{2}{m^2} \left({x^m}\right)^T D x^m\leq \frac{2}{m^2} g(x^m)\leq \frac{2}{m^2}g(x^*)$. The rest of the analysis is similar, and we obtain the same approximation factor of $\left(1-O\left(\frac{\log k}{k}\right)\right)$.
\end{remark}

To conclude, we prove that $\MSD$ remains $\classNP$-hard on distance spaces of negative type, even for a cardinality constraint. For this, we give a reduction from Densest $k$-Subgraph (D$k$S), which is $\classNP$-hard. An instance of D$k$S consists of a graph $G=(V,E)$ and a number $k\leq n=|V|$, and the object is to find a $k$-set $W\subseteq V$ whose induced subgraph $G[W]$ contains the largest number of edges. Now, the distance function $d':V^2\to \R_{\geq 0}$ defined by $d'(i,j)=2$ if $\{i,j\}\in E$, 1 if $\{i,j\}\not\in E$ is metric.\footnote{Any distance space, where the distance between distinct points is either 1 or 2, is metric.} Thus, by assertion~\ref{item:3}) (below Theorem~\ref{thr:1}), we conclude that the distance $d=(d')^{\log_2\frac{n}{n-1}}$ is of negative type, where $d(i,j)=1+\frac{1}{n-1}$ if $\{i,j\}\in E$, 1 otherwise. Finally, it is evident that an exact solution to the $\MSD$ instance $(V,d)$ with cardinality constraint $k$ corresponds to an exact solution to the D$k$S instance. This proves the next theorem. 

\begin{theorem}
  \label{thr:7}
  Max-sum dispersion $\MSD$
  is $\classNP$-hard,
  even for distance spaces of negative type and $Ax \leq b$
  representing one cardinality constraint.
\end{theorem}

\end{document}